\documentclass[runningheads,a4paper]{llncs}

\usepackage{amssymb}
\usepackage{graphicx}
\usepackage[T1]{fontenc}
\usepackage{amsmath}
\usepackage{mathbbold}
\usepackage[ruled]{algorithm2e}

\usepackage{url}
\urldef{\mailsa}\path|{nengkunyu@gmail.com, Mingsheng.Ying@uts.edu.au}|
\newcommand{\keywords}[1]{\par\addvspace\baselineskip
\noindent\keywordname\enspace\ignorespaces#1}

\def\>{\ensuremath{\rangle}}
\def\<{\ensuremath{\langle}}

\def\h{\ensuremath{\mathcal{H}}}

\newtheorem{thm}{Theorem}
\newtheorem{cor}{Corollary}
\newtheorem{lem}{Lemma}
\newtheorem{defn}{Definition}

\newcommand{\ket}[1]{|#1\rangle}

\newcommand{\ip}[2]{\langle #1|#2\rangle}
\newcommand{\op}[2]{|#1\rangle \langle #2|}

\newcommand{\tr}{{\rm tr}}
\newcommand {\spa } {{\rm span}}
\newcommand {\supp } {{\rm supp}}

\newcommand {\E } {{\mathcal{E}}}
\newcommand {\F } {{\mathcal{F}}}
\newcommand{\hs}{\mathcal{H}}

\begin{document}

\mainmatter

\title{Reachability and Termination Analysis of Concurrent Quantum Programs}

\titlerunning{Concurrent Quantum Programs}

\author{Nengkun Yu\and Mingsheng Ying}
\authorrunning{N. K. Yu and M. S. Ying}

\institute{Tsinghua University, China\\ University of Technology,
Sydney, Australia\\
\mailsa
}

\toctitle{Concurrent Quantum Programs}
\tocauthor{Yu and Ying}
\maketitle
\begin{abstract}
We introduce a Markov chain model of concurrent quantum programs. This model is a quantum generalization of Hart, Sharir and Pnueli's probabilistic concurrent programs. Some characterizations of the reachable space, uniformly repeatedly reachable space and termination of a concurrent quantum program are derived by the analysis of their mathematical structures. Based on these characterizations, algorithms for computing the reachable space and uniformly repeatedly reachable space and for deciding the termination are given.
\end{abstract}

\section{Introduction}

Research on concurrency in quantum computing started about 10 years ago, and it was motivated by two different requirements:\begin{itemize}\item \emph{Verification of quantum communication protocols}:
Quantum communication systems are already commercially available
from Id Quantique, MagiQ Technologies, SmartQuantum and NEC. Their advantage over classical communication is that
security is provable based on the principles of quantum mechanics. As is well known, it is very difficult to guarantee correctness of even classical
communication protocols in the stage of design. Thus, numerous techniques for verifying classical
communication protocols have been developed. Human intuition is much better adapted to the
classical world than the quantum world. This will make quantum protocol designers to commit
many more faults than classical protocol designers. So, it is even more
critical to develop formal methods for verification of quantum protocols (see for example~\cite{GPN08}, \cite{GPN10}, \cite{DG11}). Concurrency is a feature that must be encompassed into the formal models of quantum communication systems.

\item \emph{Programming for distributed quantum computing}: A major reason for distributed quantum computing, different from the
classical case, comes from the extreme difficulty of the physical implementation of functional
quantum computers (see for example~\cite{CE99}, \cite{YF09}). Despite convincing laboratory demonstrations of quantum computing devices,
it is beyond the ability of the current physical technology to scale them. Thus, a natural idea is to
use the physical resources of two or more small capacity quantum computers to simulate a large
capacity quantum computer. In fact, various experiments in the physical implementation of
distributed quantum computing have been frequently reported in recent years. Concurrency naturally arises in the studies of programming for distributed quantum computing.
\end{itemize}

The majority of work on concurrency in quantum computing is based on process algebras~\cite{JL04},~\cite{LJ04a},~\cite{GN05}, \cite{GN06}, \cite{La06}, ~\cite{FDJY07},~\cite{FD09},~\cite{FDY11},~\cite{DAV11}. This paper introduces a new model of concurrent quantum programs in terms of quantum Markov chains. This model is indeed a quantum extension of Hart, Sharir and Pnueli's model of probabilistic concurrent programs~\cite{HSP83}, \cite{SPH84}. Specifically, a concurrent quantum program consists of a finite set of processes. These processes share a state Hilbert space, and each of them is seen as a quantum Markov chain on the state space. The behaviour of each processes is described by a super-operator. This description of a single process follows Selinger, D'Hont and Panangaden's pioneering works~\cite{Se04}, \cite{DP06} on sequential quantum programs where the denotational semantics of a quantum program is given as a super-operator.
The super-operator description of sequential quantum programs was also adopted in one of the authors' work on quantum Floyd-Hoare logic \cite{Yi11}.
Similar to the classical and probabilistic cases~\cite{HSP83}, an execution path of a concurrent quantum program is defined to be an infinite sequence of the labels of their processes, and a certain fairness condition is imposed on an execution path to guarantee that all the processes fairly participate in a computation.

Reachability and termination are two of the central problems in program analysis and verification. The aim of this paper is to develop algorithms that compute the reachable states and decide the termination, respectively, of a concurrent quantum program. To this end, we need to overcome two major difficulties, which are peculiar to the quantum setting and would not arise in the classical case:

\begin{itemize}\item The state Hilbert space of a quantum program is a continuum and thus doomed-to-be infinite even when its dimension is finite. So, a brute-force search is totally ineffective although it may works well to solve a corresponding problem for a classical program. We circumvent the infinity problem of the state space by finding a finite characterization for reachability and termination of a quantum program through a careful analysis of the mathematical structure underlying them.
\item The super-operators used to describe the behaviour of the processes are operators on the space of linear operators on the state space, and they are very hard to directly manipulate. In particular, algorithms for computing super-operators are lacking. We adopt a kind of matrix representation for super-operators that allows us to conduct reachability and termination analysis of quantum programs by efficient matrix algorithms.
\end{itemize}

The paper is organized as follows. For convenience of the reader we briefly recall some basic notions
from quantum theory and fix the notations in Sec.~\ref{Pre}; but we refer to \cite{NC00} for more details. A Markov chain model of concurrent quantum programs is defined in Sec.~\ref{Mod}, where we also give a running example of quantum walks. In Sec.~\ref{RRE}, we present a characterization for reachable space and one for uniformly repeatedly reachable space of a quantum program, and develop two algorithms to compute them. A characterization of termination of a quantum program with fair execution paths and an algorithm for deciding it are given in Sec.~\ref{TER}. It should be pointed out that termination decision in Sec.~\ref{TER} is based on reachability analysis in Sec.~\ref{RRE}. A brief conclusion is drawn in Sec.~\ref{CCO}.

\section{Preliminaries and Notations}\label{Pre}
\subsection{Hilbert Spaces}
The state space of a quantum system is a Hilbert space.
In this paper, we only consider a finite dimensional Hilbert space $\hs$, which is a complex vector
space equipped with an inner product $\langle\cdot|\cdot\rangle$. A pure state of a quantum system is represented by a unit vector,
i.e., a vector $|\psi\rangle$ with $\langle\psi|\psi\rangle=1$.
Two vectors $|\varphi\rangle,|\psi\rangle$ in $\hs$ are orthogonal, written $|\varphi\rangle\perp|\psi\rangle$, if
their inner product is $0$. A basis of $\hs$ is orthonormal if its elements are mutually orthogonal, unit vectors. The trace of a linear operator $A$ on $\hs$ is defined to be
$tr(A)=\sum_{i}\langle i|A|i\rangle$, where $\{|i\rangle\}$ is an orthonormal basis of $\hs$. For a subset $V$ of $\hs$, the subspace $\spa V$ spanned by $V$ consists of all linear combinations
of vectors in $V$. For any subspace $X$ of $\hs$, its orthocomplement
is the subspace $X^{\bot}=\{|\varphi\rangle\in \mathcal{H}:|\varphi\rangle\perp|\psi\rangle\
{\rm for\ all}\ |\psi\rangle\in X\}$. The join of a family $\{X_i\}$ of subspaces is $\bigvee_i
X_i=\spa(\bigcup_i X_i).$ In particular, we write $X\vee Y$ for the join of
two subspaces $X$ and $Y$.
A linear operator $P$ is called the projection onto a subspace $X$ if
$P\ket{\psi}=\ket{\psi}$ for all $\ket{\psi}\in X$ and $P\ket{\psi}=0$
for all $\ket{\psi}\in X^{\bot}$. We write $P_X$ for the projection onto $X$.

A mixed state of a quantum system is represented by a density operator.
A linear operator $\rho$ on $\hs$ is called a density operator
(resp. partial density operator) if $\rho$ is positive-semidefinite in the sense that $\langle
\phi|\rho|\phi\rangle \geq 0$ for all $|\phi\rangle$, and
$tr(\rho)=1$ (resp. $tr(\rho)\leq 1$). For any statistical
ensemble $\{(p_i,|\psi_i\rangle)\}$ of pure quantum states with $p_i> 0$
for all $i$ and $\sum_i p_i=1$,  $\rho=\sum_{i}p_i|\psi_i\rangle\langle\psi_i|$
is a density operator. Conversely, each density operator can be generated by
an ensemble of pure states in this way. In particular, we write $\psi$ for the
density operator $\op{\psi}{\psi}$ generated by a single pure states $\ket{\psi}$.
The support of a partial density operator $\rho$, written $\supp (\rho)$, is the space spanned by its eigenvectors with nonzero eigenvalues.

\begin{lem}\label{prel-1} For any $p>0$ and partial density operators $\rho, \sigma$, we have: (1) $\supp(p\rho)=\supp(\rho)$;
(2) $\supp(\rho)\subseteq\supp(\rho+\sigma)$; (3) $\supp(\rho+\sigma)=\supp(\rho)\vee\supp(\sigma)$.
\end{lem}

\subsection{Super-Operators}

A super-operator is a mathematical formalism used to describe a  broad class of transformations that a quantum system can undergo.
 A super-operator on $\hs$ is a linear operator
$\E$ from the space of linear operators on $\hs$ into itself, satisfying (1) $tr[\E(\rho)]\leq tr(\rho)$ for any $\rho$; (2) Complete positivity(CP): for any extra Hilbert
space $\hs_k$, $(\mathcal{I}_k\otimes \E)(A)$ is positive
provided $A$ is a positive operator on $\hs_k\otimes \hs$, where
$\mathcal{I}_k$ is the identity operation on $\hs_k$.
Furthermore, if $tr[\E(\rho)]=tr(\rho)$
for any $\rho$, then $\E$ is said
to be trace-preserving. Each super-operator $\E$ enjoys the Kraus representation: there exists a set of
operators $\{E_i\}$ satisfying (1) $\E(\rho)=\sum_{i}E_i\rho E_i^{\dag}$ for all density
operators $\rho$; (2) $\sum_{i}E_i^{\dag}E_i\leq I$, with
equality for trace-preserving $\E$, where $I$ is the
identity operator. In this case, we write $\E=\sum_i E_i\cdot E_i^{\dag}$.
The image of subspace $X$ of $\hs$ under $\E$ is $\E(X)=\bigvee_{\ket{\psi}\in X}\supp (\E(\psi)),$ and
the pre-image of $X$ under $\E$ is $\E^{-1}(X)=\{\ket{\psi}\in\hs:\supp
(\E(\psi))\subseteq X\}.$

\begin{lem}\label{prel-2} (1) $\supp(\rho)\subseteq\supp(\sigma)\Rightarrow\supp(\E(\rho))\subseteq \supp(\E(\sigma))$, and $\supp(\rho)=\supp(\sigma)\Rightarrow\supp(\E(\rho))= \supp(\E(\sigma))$.

(2) $\supp(\E(\rho))\subseteq \supp((\E+\F)(\rho))$. \ \ \ (3) $\E(X)=\supp(\E(P_X))$.

(4)  $X\subseteq Y\Rightarrow \E(X)\subseteq \E(Y)$. \ \ \ \ \ \ \ \ \ \ \ \ (5)  $\E(X)\subseteq(\E+\F)(X)$.

(6) If $\E=\sum_{i}E_i\cdot E_i^{\dag}$, then $\E^{-1}(X)=[\supp(\mathcal{E}^{\ast}(P_{X^{\bot}}))]^{\bot},$
where $\mathcal{E}^{\ast}=\sum_i E_i^{\dag} \cdot E_i$ is
the (Schr\"odinger-Heisenberg) dual of $\E$.\end{lem}

\subsection{Matrix Representation of Super-Operator}
The matrix representation of a super-operator is usually easier in \cite{YYFD11} to
manipulate than the super-operator itself. If $\mathcal{E}=\sum_iE_i\cdot E_i^{\dag}$ and $\dim \h=d$, then the matrix representation of
$\mathcal{E}$ is the $d^2\times d^2$ matrix $M=\sum_i
E_i\otimes E_i^*,$ where $A^{\ast}$ stands for the conjugate of
matrix $A$, i.e., $A^{\ast}=(a^{\ast}_{ij})$ with $a^{\ast}_{ij}$
being the conjugate of complex number $a_{ij}$, whenever
$A=(a_{ij})$. According to \cite{YYFD11}, we have the following

\begin{lem}\label{prel-3} (1)
The modulus of any eigenvalue of $M$ is less or equal to 1.

(2) We write
$|\Phi\rangle=\sum_j|jj\rangle$ for the (unnormalized) maximally
entangled state in $\h\otimes\h$, where $\{|j\rangle\}$ is an
orthonormal basis of $\hs$. Then for any $d\times d$ matrix $A$,
we have $(\mathcal{E}(A)\otimes I)|\Phi\rangle=M(A\otimes
I)|\Phi\rangle.$\end{lem}

\subsection{Quantum Measurements}

A quantum measurement is described by a collection $\{M_m\}$ of operators, where the indexes $m$ refer to the
measurement outcomes. It is required that the measurement operators
satisfy the completeness equation $\sum_{m}M_m^{\dag}M_m=I_{\hs}$. If
the system is in state $\rho$, then the probability that measurement
result $m$ occurs is given by $p(m)=tr(M_m^{\dag}M_m\rho)$, and the
state of the system after the measurement is $\frac{M_m\rho
M_m^{\dag}}{p(m)}.$

\section{A Model of Concurrent Quantum Programs}\label{Mod}
Our model is a
quantum extension of Hart, Sharir and Pnueli's probabilistic concurrent programs \cite{HSP83}.
A concurrent quantum program consists of a finite set $K=\{1,2,\cdots,m\}$ of
quantum processes, and these processes have a common state space,
which is assumed to be a $d$-dimensional Hilbert space $\hs$.
With each $k\in K$ we associate a trace-preserving super-operator
$\E_k$, describing a single atomic action or evolution of process $k$.
 Also, we assume a termination condition for the program. At the
end of each execution step, we check whether this condition is satisfied or not.
The termination condition is modeled by a yes-no measurement
$\{M_0,M_1\}$: if the measurement outcome is $0$, then the
program terminates, and we can imagine the program state falls into
a terminal (absorbing) space and it remains there forever; otherwise, the
program will enter the next step and continues to perform a
quantum operation chosen from $K$.
\begin{defn}\label{p-def}A concurrent quantum program defined on a $d$-dimensional Hilbert space $\hs$ is a pair
$\mathcal{P}=(\{\E_k:k\in K\}, \{M_0,M_1\}),$ where:
\begin{enumerate}
\item $\E_k$ is a super-operator on $\hs$ for each $k\in K$;
\item $\{M_0, M_1\}$ is a measurement on $\hs$ as the termination test.
\end{enumerate}
\end{defn}
Any finite string $s_1 s_2\cdots s_m$ or infinite string $s_1 s_2\cdots s_i\cdots$
of elements of $K$ is called a execution path of the program. Thus, the sets of finite
and infinite execution paths of program $\mathcal{P}$ are \begin{eqnarray*}
S&=&K^{\omega} =\{s_1 s_2\cdots s_i\cdots:s_i\in K\ {\rm for\ every}\ i\geq 1\},\\
S_{fin}&=&K^*=\{s_1 s_2\cdots s_m:m\geq 0\ {\rm and}\ s_i\in K\ {\rm for\ all}\ 1\leq i\leq m\},
\end{eqnarray*}respectively. A subset of $S$ is usually called a schedule.

For simplicity of presentation, we introduce the
notation $\mathcal{F}_k$ for any $k\in K$ which stands for the super-operator defined
by
$\mathcal{F}_{k}(\rho)=\E_k(M_1\rho M_1^\dagger)$
for all density operators $\rho$.
Assume the initial state is $\rho_0$.
The execution of
the program under path $s=s_1 s_2\cdots s_k\cdots\in S$ can be described as follows. At the first
step, we perform the termination measurement $\{M_0,M_1\}$ on the
initial state $\rho_0$. The probability that the program terminates;
that is, the measurement outcome is $0$, is
$\tr[M_0\rho_0M_0^{\dag}]$.  On the other hand, the probability that the program
does not terminate; that is, the measurement outcome is $1$, is
$p_1^s=\tr[M_1\rho_0M_1^{\dag}],$ and the program state after the
outcome $1$ is obtained is
$\rho_1^s=M_1\rho_0M_1^{\dag}/p^s_1.$ We adopt
Selinger's normalization convention~\cite{Se04} to encode
probability and density operator into a partial
density operator
$p_1^s\rho_1^s=M_1\rho_0M_1^{\dag}.$ Then this (partial) state is transformed
by the quantum operation $\E_{s_1}$ to $\E_{s_1}(M_1\rho_0M_1^{\dag})=\mathcal{F}_{s_1}(\rho_0)$.
The program continues its computation step by step according to the path $s$.
In general, the $(n+1)$th step is executed
upon the partial density operator
$p^s_{n}\rho^s_{n}=\mathcal{F}_{s_n}\circ\cdots\circ
\mathcal{F}_{s_2}\circ \mathcal{F}_{s_1}(\rho_0),$ where $p^s_{n}$ is the probability
that the program does not terminate at the $n$th step, and
$\rho^s_{n}$ is the program state after the termination measurement
is performed and outcome $1$ is reported at the $n$th step.
For simplicity, let $\mathcal{F}_f$ denote the super-operator $\mathcal{F}_{s_n}\circ\cdots\circ
\mathcal{F}_{s_2}\circ \mathcal{F}_{s_1}$ for string $f=s_1s_2\cdots s_n$.
Thus, $p^s_n\rho^s_{n}=\mathcal{F}_{s[n]}(\rho_0),$
where $s[n]$ is used to denote the head $s_1s_2\cdots s_n$ for any $s=s_1s_2\cdots s_n\cdots\in S$.
The probability that the program
terminates in the $(n+1)$th step is then
$\tr(M_0(\mathcal{F}_{s[n]}(\rho_0))M_0^{\dag}),$ and the
probability that the program does not terminate in the $(n+1)$th
step is $p^s_{n+1}=\tr(M_1(\mathcal{F}_{s[n]}(\rho_0))M_1^{\dag}).$

\subsection{Fairness}
To guarantee that all the processes in a concurrent program can fairly participate in a computation, a certain fairness condition on its execution paths is needed.

\begin{defn}An infinite execution path $s=s_1s_2...s_i...\in S$ is fair
if each process appears infinitely often in $s$; that is, for each $k\in K$,
there are infinitely many $i\geq 1$ such that $s_i=k$.
\end{defn}
We write $F=\{s:s\in S~~\rm{is ~~fair}\}$ for the schedule of all fair execution paths.

\begin{defn}
A finite execution path $\sigma=s_1s_2\cdots s_n\in S_{fin}$ is called a
 \textit{fair piece} if each process appears during $\sigma$;
 that is, for each $k\in K$, there exists $i\leq n$ such that $s_i=k$.\end{defn}
$F_{fin}$ is used to denote the set of all fair pieces:
$F_{fin}=\{\sigma:\sigma\in S_{fin}~\rm{is~ a~fair~piece}\}.$
It is obvious that $F=F_{fin}^{\omega};$ in other words,
every fair infinite execution path $s\in F$ can be divided into
an infinite sequence of fair pieces: $s=f_1f_2\cdots f_k\cdots$, where $f_i\in F_{fin}$ for each $i>0$.
The fairness defined above can be generalized by introducing the notion
of fairness index, which measures the occurrence frequency of every
 process in an infinite execution path.
\begin{defn}\label{ind-def}
For any infinite execution path $s\in F$, its fairness index $f(s)$ is the minimum,
 over all processes, of the lower limit of the occurrence frequency of the processes in $s$; that is,
  $$f(s)=\min_{k\in K} \lim_{t\rightarrow \infty} \inf_{n>t} \frac{s(n,k)}{n},$$
where $s(n,k)$ is the number of occurrences of $k$ in $s[n]$.\end{defn}
For any $\delta\geq 0$, we write $F_{\delta}$ for the set of infinite
execution paths whose fairness index is greater than $\delta$:
$F_{\delta}=\{s:s\in S\ {\rm and}\ f(s)>\delta\}.$
Intuitively, within an infinite execution path in $F_{\delta}$,
each process will be woken up with frequency greater than $\delta$. It is clear that
$F_0\subsetneq F.$

\subsection{Running Example}\label{runes}
We consider two quantum walks on a circle $C_3=(V,E)$ with vertices $V=\{0,1,2\}$ and edges $E=\{(0,1), (1,2),(2,0)\}$.
The first quantum walk $\mathcal{W}_1=(\{W_1\},\{M_0,M_1\})$ is given as follows:
\begin{itemize}
\item The state space is the $3-$dimensional Hilbert space with computational basis $\{\ket{i}|i\in
V\}$;
\item The initial state is $\ket{0}$; this means that the walk starts at the vertex $0$;
\item A single step of the walk is defined by the unitary operator:
$$W_1=\frac{1}{\sqrt{3}}\left(\begin{array}{ccc}
1 & 1 & 1 \\
1 & w & w^2\\
1 & w^2 & w
\end{array}\right),$$
where $w=e^{2\pi i/3}$.
Intuitively, the probabilities of walking to the left
and to the right are both $1/3$, and there is also a
probability $1/3$ of not walking.
\item The termination measurement $\{M_0,M_1\}$ is defined by $$M_0=\op{2}{2},\ M_1=I_3-\op{2}{2},$$
where $I_3$ is the $3\times 3$ unit matrix.
\end{itemize} The second walk $\mathcal{W}_2=(\{W_2\},\{M_0,M_1\})$ is similar to the first one, but its single step is described by unitary operator
$$W_2=\frac{1}{\sqrt{3}}\left(\begin{array}{ccc}
1 & 1 & 1\\
1 & w^2 & w\\
1 & w & w^2
\end{array}\right).$$Then we can put these two quantum walks together to form a concurrent program $\mathcal{P}=(\{W_1,W_2\},\{M_0,M_1\})$. For example,
the execution of this concurrent program according to unfair path $1^{\omega}\notin F$ is equivalent to a sequential program $(\{W_1\},\{P_0,P_1\})$;
and the execution of $\mathcal{P}$ according to fair path $(12)^{\omega}\in F$ is as follows: we perform the termination measurement $\{M_0,M_1\}$ on the
initial state $\rho_0$, then the nonterminating part of the program state is transformed by
the super-operator $\mathcal{U}_1=W_1\cdot W_1^{\dag}$, followed by the termination measurement, and then the application of the super-operator $\mathcal{U}_2=W_2\cdot W_2^{\dag}$, and this procedure is repeated infinitely many times.

\section{Reachability}\label{RRE}

Reachability is at the centre of program analysis. A state is reachable if some finite execution starting in the initial state ends in it. What concerns us in the quantum case is the subspace of $\mathcal{H}$ spanned by reachable states.
\begin{defn}
The reachable space of program $\mathcal{P}=(\{\E_k:k\in K\}, \{M_0,M_1\})$
starting in the initial state $\rho_0$ is $$\hs_{R}=\bigvee_{s\in S,j\geq 0}\supp~\mathcal{F}_{s[j]}(\rho_0)=\bigvee_{f\in S_{fin}}\supp~\mathcal{F}_f(\rho_0).$$
\end{defn}
We have the following closed form characterization of the reachable space.
\begin{thm}\label{reachablem}
$\label{xxx}\hs_R=\supp(\sum_{i=0}^{d-1} \F^i(\rho_0)),$ where $d=\dim\mathcal{H}$ is the dimension of $\h$,
and $\F=\sum_{k\in K}\mathcal{F}_k$.
\end{thm}
\textit{Proof:} We write $X$ for the right-hand side. From Lemma~\ref{prel-1}, we see that
$X=\bigvee\{\supp \mathcal{F}_f(\rho_0):  f\in S_{fin}, |f|<d\},$
where $|f|$ denotes the length of string $f$.
According to the definition of reachable space, we know that $X\subseteq \h_{R}$.
To prove the inverse part $X\supseteq \h_{R}$,
for each $n\geq 0$, we define subspace $Y_n$  as follows:
$Y_n:=\supp(\sum_{i=0}^n \F^i(\rho_0)).$
Due to Lemma~\ref{prel-1}, we know that
$Y_0\subseteq Y_1\subseteq \cdots \subseteq Y_n\subseteq\cdots.$
Suppose $r$ is the smallest integer satisfying $Y_r=Y_{r+1}$.
We observe that $Y_{n+1}=\supp(\rho+\F(P_{Y_{n}}))$ for all $n\geq 0$. Then it follows that $Y_n=Y_r$ for all $n\geq r$. On the other hand,
we have $Y_0\subsetneq Y_1\subsetneq\cdots \subsetneq Y_r.$
So, $0<d_0<d_1<\cdots<d_r\leq d,$ where $d_0$ is the rank of $\rho_0$,
and $d_i$ is the dimension of subspace $Y_i$ for $1\leq i\leq r$.
Therefore, we have $r\leq d-1$ and $Y_{d-1}=Y_r\supseteq Y_n$ for all $n$. Finally, for any $f\in S_{fin}$, it follows from Lemma~\ref{prel-2} that
$\supp(\mathcal{F}_f(\rho_0))\subseteq \supp(\F^{|f|}(\rho_0))\subseteq Y_{|f|}\subseteq Y_{d-1}=X.$
Thus, $\h_{R}\subseteq X$. \hfill $\blacksquare$
Now we are able to present an algorithm computing reachable subspace using matrix representation of super-operators. We define $\mathcal{G}=\sum_{k\in K}\mathcal{F}_k/|K|$.

\begin{algorithm}
\caption{Computing reachable space \label{alg:Rs}}
\SetKwInOut{Input}{input}\SetKwInOut{Output}{output}
\Input{An input state $\rho_0$, and the matrix representation $G$ of
$\mathcal{G}$}
\Output{An orthonormal basis $B$ of $\mathcal{H}_{R}$.}
$\ket{x}\leftarrow (I-G/2)^{-1}(\rho_0\otimes I)|\Phi\rangle$\;
(* $|\Phi\rangle=\sum_j|j_Aj_B\rangle$ is the unnormalized maximally entangled
state in $\h\otimes\h$  *)\\
\For {$j=1:d$}{
$\ket{y_j}\leftarrow\langle j_B|x\rangle;$
}
\textbf{set of} states $B\leftarrow\emptyset$\;
\textbf{integer} $l\leftarrow~0$\; 
\For {$j=1:d$}{
$\ket{z}\leftarrow \ket{y_j}-\sum_{k=1}^{l}\ip{b_k}{y_j}\ket{b_k}$\;
\If{$\ket{z}\neq 0$}{
$l\leftarrow~l+1$\; 
$\ket{b_l}\leftarrow \ket{z}/\sqrt{\ip{z}{z}}$\;
$B\leftarrow B\cup\{\ket{b_l}\}$\;
}
}
\Return
\end{algorithm}
\begin{thm}\label{Al-1}Algorithm 1 computes the reachable space in time  $\mathcal{O}(d^{4.7454})$ , where $d=\dim\mathcal{H}$.
\end{thm}
\textit{Proof:} It follows from Lemma \ref{prel-3}(1) that $I-G/2$ is invertible, and $\sum_{i=0}^\infty (G/2)^i=(I-G/2)^{-1}$.
We write $\rho=\sum_{i=0}^\infty \mathcal{G}^i(\rho_0)/2^i$, and have
$$(\rho\otimes I)\ket{\Phi}=\sum_{i=0}^\infty (G/2)^i(\rho_0\otimes I)\ket{\Phi},$$ and the existence of $\rho$ immediately follows from Lemma \ref{prel-3}. We further see that $\ket{x}=(\rho\otimes I)\ket{\Phi}=\sum_{j}\rho\ket{j_A}\ket{j_B}$ and $\ket{y_i}=\rho\ket{j_A}$. Note that $B$ is obtained from $\{|y_j\rangle\}$ by the Gram-Schmidt procedure. So, $\supp(\rho)=\spa \{\rho\ket{j}\}=\spa B$.
It is clear that $
\mathcal{H}_R=\supp(\sum_{i=0}^{d-1}\mathcal{F}^i(\rho_0))\subseteq\supp(\rho)$. Therefore, $\mathcal{H}_R=\supp(\rho)=\spa B$, and the algorithm is correct.

The complexity comes from three the following parts: (1) it costs $\mathcal{O}(d^{2*2.3727})$ to compute $(I-G/2)^{-1}$
by using Coppersmith-Winograd algorithm \cite{DS90}; (2) it requires $\mathcal{O}(d^{4})$ to obtain $\ket{x}$ from $(I-G/2)^{-1}$;
(3) the Gram-Schmidt orthonormalization is in time $\mathcal{O}(d^{3})$.
So, the time complexity is $\mathcal{O}(d^{4.7454})$ in total. \hfill $\blacksquare$

An advantage of Algorithm 1 is that we can store $(I-G/2)^{-1}$.
Then for any input state $\rho_0$, we only need $\mathcal{O}(d^{4})$
to compute the space reachable from $\rho_0$.

\begin{defn}
The uniformly repeatedly reachable space of program $\mathcal{P}=(\{\E_k:k\in K\}, \{M_0,M_1\})$
starting in the initial state $\rho_0$ is
$$\hs_{URR}=\bigcap_{n\geq 0}\bigvee_{s\in S,j\geq n}\supp~\mathcal{F}_{s[j]}(\rho_0)=\bigcap_{n\geq 0}\bigvee\{\supp~\mathcal{F}_f(\rho_0):f\in S_{fin}, |f|\geq n\}.$$
\end{defn}
The uniformly repeatedly reachable space enjoys the following closed form,
\begin{thm}\label{URR}
$\hs_{URR}=\supp(\sum_{i=d}^{2d-1} \F^i(\rho_0)),$
where $d=\dim \h$, and $\F=\sum_{k\in K}\mathcal{F}_k$.
\end{thm}
\textit{Proof:} For each $n\geq 0$, we define subspace $Z_n$ as follows: $Z_n:=\bigvee_{j\geq n}\supp~\mathcal{F}^j(\rho_0).$
It is obvious that
$Z_0\supseteq Z_1\supseteq \cdots \supseteq Z_n\supseteq\cdots.$ Suppose $r$ is the smallest integer satisfying $Z_r=Z_{r+1}$. By noting that $Z_{n+1}=\supp(\F(P_{Z_{n}}))$, we can show that $Z_n=Z_r$ for all $n\geq r$. On the other hand, we have $Z_0\supsetneq Z_1\supsetneq\cdots \supsetneq Z_r.$
So, $d_0>d_1>\cdots>d_r\geq 0,$
and $d_i$ is the dimension of subspace $Z_i$ for $0\leq i\leq r$.
Therefore, we have $r\leq d_0\leq d$ and $Z_{d}=Z_r$. Therefore, $\hs_{URR}=\bigcap_{n\geq 0}Z_n=Z_d$.
It is obvious that $Z_d$ is the reachable space starting in state $\mathcal{F}^d(\rho_0)$. Using Theorem~\ref{reachablem} we obtain
$Z_d=\supp(\sum_{i=0}^{d-1}\mathcal{F}^i(\mathcal{F}^d(\rho_0)))=\supp(\sum_{i=d}^{2d-1}\mathcal{F}(\rho_0)).$ \hfill $\blacksquare$
We can give an algorithm computing the uniformly repeatedly reachable space by combining the above theorem and matrix representation of super-operators.
\begin{algorithm}
\caption{Compute uniformly repeatedly reachable space \label{alg:URR}}
\SetKwInOut{Input}{input}\SetKwInOut{Output}{output}
\Input{An input state $\rho_0$, and the matrix representation $G$ of
$\mathcal{G}$}
\Output{An orthonormal basis $B_{URR}$ of $\mathcal{H}_{URR}$.}
$\ket{x}\leftarrow G^d(I-G/2)^{-1}(\rho_0\otimes I)|\Phi\rangle$\;
(* $|\Phi\rangle=\sum_j|j_Aj_B\rangle$ is the unnormalized maximally entangled
state in $\h\otimes\h$  *)\\
\For {$j=1:d$}{
$\ket{y_j}\leftarrow\langle j_B|x\rangle;$
}
\textbf{set of} states $B_{URR}\leftarrow\emptyset$\;
\textbf{integer} $l\leftarrow~0$\; 
\For {$j=1:d$}{
$\ket{z}\leftarrow \ket{y_i}-\sum_{k=1}^{l}\ip{b_k}{y_j}\ket{b_k}$\;
\If{$\ket{z}\neq 0$}{
$l\leftarrow~l+1$\; 
$\ket{b_l}\leftarrow \ket{z}/\sqrt{\ip{z}{z}}$\;
$B_{URR}\leftarrow B_{URR}\cup\{\ket{b_l}\}$\;
}
}
\Return
\end{algorithm}
\begin{thm}Algorithm 2 computes the uniformly repeatedly reachable space in time  $\mathcal{O}(d^{4.7454}\log d)$, where $d=\dim\hs$.\end{thm}
\textit{Proof:} This theorem is a corollary of Theorem~\ref{Al-1}. Here, $\log d$ in the complexity comes from computing $M^d$ using the method of exponentiation by squaring.

\section{Termination}\label{TER}

Another important problem concerning the behaviour of a program is its termination.
\begin{defn}
Let the program $\mathcal{P}=(\{\E_k:k\in K\}, \{M_0,M_1\})$. Then $\mathcal{P}$ with input $\rho_0$ terminates for execution
path $s\in S$ if $p^s_{n}\rho^s_{n}=\mathcal{F}_{s[n]}=0$ for some positive integer $n$.
\end{defn}

\begin{defn}\begin{enumerate}
\item If a program $\mathcal{P}$ with input $\rho_0$ terminates for all $s\in A$,
then we say that it terminates in schedule $A$.
\item If there is a positive integer $n$ such that $p^s_{n}\rho^s_{n}=0$ for all $s\in A$,
then it is said that the program $\mathcal{P}$ with input $\rho_0$ uniformly terminates in schedule $A$.
\end{enumerate}
\end{defn}

We first prove the equivalence between termination and uniform termination. Of course, this equivalence comes from finiteness of the dimension of the state space.

\begin{thm}\label{termination}
The program $\mathcal{P}=(\{\E_k:k\in K\}, \{M_0,M_1\})$ with initial state $\rho_0$
terminates in the biggest schedule $S=K^{\omega}$ if and only if it uniformly terminates in schedule $S$.
\end{thm}
\begin{proof} The \textquotedblleft if\textquotedblright\ part is obvious. We prove the  \textquotedblleft only if\textquotedblright\ part in two steps:

(1) We consider the case of $|K|=1$, where $\{\E_k:k\in K\}$ is a singleton $\{\E\}$. Now the program is indeed a sequential program, and it is a quantum loop~\cite{YF10}. We
write $\F(\rho)=\E (M_1\rho M_1^{\dag})$ for all $\rho$. What we need to prove is that if $\mathcal{P}$ terminates, i.e., $\F^n(\rho_0)=0$ for some $n$,
then it terminates within $d$ steps, i.e., $\F^d(\rho_0)=0$. If $\rho_0$ is a pure state $|\psi\rangle$, then we define the termination sets as follows: $X_n:=\{\ket{\psi}:\F^n(\psi)=0\}$ for each integer $n>0$.

(1.1) If $\ket{\varphi},\ket{\chi}\in X_n$, then $\F^n(\varphi+\chi)=0$,
which leads to $\alpha\ket{\varphi}+\beta\ket{\chi}\in X_n$ for any $\alpha,\beta
\in \mathbb{C}$. Thus $X_n$ is a subspace of $\hs$.

(1.2) Since $\F^n(\psi)=0\Rightarrow \F^{n+1}(\psi)=0$,
it holds that that $X_n\subseteq X_{n+1}$ for any $n>0$.
So, we have the inclusion relation $X_1\subseteq X_2\subseteq \cdots \subseteq X_n\subseteq\cdots.$

Now suppose $t$ is the smallest integer satisfying $X_t=X_{t+1}$.
Invoking Lemma \ref{prel-2}, we obtain that
 $\supp({\F^{\ast}}^t(I))=X_t^{\bot}=X_{t+1}^{\bot}=\supp({\F^{\ast}}^{t+1}(I)),$
 where $\F^{\ast}(\cdot)$ denotes the (Schr\"odinger-Heisenberg) dual of $\F(\cdot)$.
We have $\supp({\F^{\ast}}^{n}(I))=\supp({\F^{\ast}}^{t}(I))$,
which leads to $X_n=X_{t}$ for all $n\geq t$. Now, it holds that
$X_1\subsetneq X_2\subsetneq \cdots \subsetneq X_t=X_{t+1}=X_{t+2}=\cdots.$
This implies $d_1<d_2\cdots< d_t,$
where $d_i$ is the dimension of subspace $X_i$.
Thus, $t\leq d$. If $\F^n(\psi)=0$, then $\ket{\psi}\in X_n\subseteq X_d$, and $\F^d(\psi)=0$.

In general, if $\rho_0$ is a mixed input state $\rho_0=\sum p_i \op{\psi_i}{\psi_i}$ with all $p_i>0$,
and $\F^n(\rho_0)=0$, then $\F^n(\psi_i)=0$ for all $i$.
Therefore,  $\F^d(\psi_i)=0$ for all $i$,
and it follows immediately that $\F^d(\rho_0)=0$.

(2) For the general case of $|K|\geq 2$, we assume that $\mathcal{P}$ starting in $\rho_0$ terminates in $S$, i.e.,
for any $s\in S$, there exists an integer $n_s$ such that $\mathcal{F}_{s[n_s]}(\rho_0)=0$.
Our purpose is to show that there exists an integer $n$ such that $\mathcal{F}_{s[n]}(\rho_0)=0$ for all $s\in S$.
Indeed, we can choose $n=d$. We do this by refutation.
Assume that $\mathcal{F}_{s[d]}(\rho_0)\neq 0$ for some $s\in S$.
We are going to construct an execution path $s\in S$ such that $\mathcal{F}_{s[n]}(\rho_0)\neq 0$ for any $n\geq 0$.
Let $\F=\sum_{k\in K} \mathcal{F}_k$. Then the assumption means that
there exist $f\in K^{d}$ such that $\mathcal{F}_f(\rho_0)\neq 0$,
and it follows that $\mathcal{F}^d(\rho_0)\neq 0.$
Now we consider the loop program $(\{\F\},\{M_0,M_1\})$ with initial state $\rho_0$.
Applying (1) to it, we obtain $\mathcal{F}^{2d}(\rho_0)\neq 0$.
Then there exist $g_1,h_1\in K^{d}$ such that $\mathcal{F}_{h_1}(\mathcal{F}_{g_1}(\rho_0))=\mathcal{F}_{g_1h_1}(\rho_0)\neq 0$, and $\mathcal{F}^{d}(\mathcal{F}_{g_1}(\rho_0))\neq 0$.
Applying (1) again leads to $\mathcal{F}^{2d}(\mathcal{F}_{g_1}(\rho_0))\neq0$,
which means that there exist $h_2,g_2 \in K^{2d}$ such that $\mathcal{F}_{h_2}(\mathcal{F}_{g_1g_2}(\rho_0))=\mathcal{F}_{g_2h_2}(\mathcal{F}_{g_1}(\rho_0))\neq0$.
Thus, we have $\mathcal{F}^d(\mathcal{F}_{g_2g_1}(\rho_0))\neq 0$.
Repeating this procedure, we can find an infinite sequence $g_1,g_2,... \in K^d$.
Put $s=g_1g_2...\in S$. Then it holds that $\mathcal{T}_{s[kd]}(\rho_0)\neq 0$ for any integer $k$.
Thus, we have $\mathcal{T}_{s[n]}(\rho)\neq0$ for all $n$. \hfill $\blacksquare$\end{proof}

Now we are ready to consider termination under fairness.
Of course, any  permutation of $K$ is a fair piece.
We write $P_K$ for the set of permutations of $K$.
For $\sigma=s_1s_2\cdots s_m\in P_K$,
a finite execution path of the form $s_1\sigma_1s_2\sigma_2\cdots \sigma_{m-1}s_m$
is called an expansion of $\sigma$.
Obviously, for any $\sigma\in P_K$, all of its expansions are in $F_{fin}$.
We will use a special class of fair pieces generated by permutations:
\begin{equation*}
\Pi=\{s_1\sigma_1s_2\sigma_2\cdots \sigma_{m-1}s_m: s_1s_2\cdots s_m\in P_K\ {\rm and}\ |\sigma_i|<d\ {\rm for\ every}\ 1\leq i<m\},
\end{equation*}
where $d$ is the dimension of the Hilbert space $\mathcal{H}$ of program states.
It is easy to see that $\Pi\subsetneq F_{fin}$.

\begin{thm}\label{ftermination}
A program $\mathcal{P}=(\{\E_k:k\in K\}, \{M_0,M_1\})$ with initial state $\rho_0$
terminates in the fair schedule $F$ if and only if it terminates in the schedule $\Pi^{\omega}$.
\end{thm}
\begin{proof} The \textquotedblleft only if\textquotedblright\ part is clear because $\Pi^{\omega}\subseteq F$.
To prove the \textquotedblleft if\textquotedblright\ part, assume $\mathcal{P}$ terminates in the schedule $\Pi^{\omega}$. We proceed in four steps:

(1) Since $\Pi$ is a finite set, we can construct a new program $\mathcal{P}^\prime=(\{\mathcal{F}_f:f\in\Pi\},\{0,I\})$. (We should point out that $\mathcal{F}_f$ is usually not trace-preserving, and thus $\mathcal{P}^\prime$ is indeed not a program in the sense of Definition~\ref{p-def}. However, this does not matter for the following arguments.) It is easy to see that the termination of $\mathcal{P}$ with $\rho_0$ in schedule $\Pi^\omega$ implies the termination of $\mathcal{P}^\prime$ with $\rho_0$ in $\Pi^\omega$. Note that $\Pi^\omega$ is the biggest schedule in $\mathcal{P}^\prime$, although it is not the biggest schedule in $\mathcal{P}$. So, we can apply Theorem~\ref{termination} to $\mathcal{P}^\prime$ and assert that
$(\sum_{f\in \Pi} \mathcal{F}_f)^d(\rho_0)=0.$
That is equivalent to
\begin{equation}\label{ft-1}\supp[(\sum_{f\in \Pi} \mathcal{F}_f)^d(\rho_0)]=\{0\}\ (0-{\rm dimensional\ subspace}).\end{equation}

(2) For each $\sigma\in P_K$, we set $A_\sigma=\{\sigma^\prime\in\Pi: \sigma^\prime\ {\rm ~is~an~expansion~of}~\sigma\}.$
Then $\bigcup\limits_{\sigma\in P_K} A_\sigma=\Pi$. Moreover, we write $\mathcal{G}_{\sigma}=\sum_{f\in A_{\sigma}} \mathcal{F}_f$ for every $\sigma\in P_K$. It is worth noting that $\sum_{\sigma\in P_K}\mathcal{G}_\sigma=\sum_{f\in\Pi}\mathcal{F}_f$ is not true in general because it is possible that $A_{\sigma_1}\cap A_{\sigma_2}\neq\emptyset$ for different $\sigma_1$ and $\sigma_2$. But by Lemma~\ref{prel-1}.1 and 3 we have
$$\supp[(\sum_{\sigma\in P_K}\mathcal{G}_{\sigma})(\rho_0)]=\supp[(\sum_{f\in \Pi} \mathcal{F}_f)(\rho_0)],$$ and furthermore, it follows from Eq.~(\ref{ft-1}) that
\begin{equation}\label{ft-2}\supp[(\sum_{\sigma\in P_K} \mathcal{G}_\sigma)^d(\rho_0)]=\supp[(\sum_{f\in \Pi} \mathcal{F}_f)^d(\rho_0)]=\{0\}.\end{equation}

(3) For each fair piece $\sigma^\prime\in F_{fin}$, and for any $\rho$, we can write $\sigma^\prime=s_1f_1s_2\cdots$ $s_{m-1}f_{m-1}s_m$ for some $\sigma_0=s_1s_2\cdots s_m\in P_K$, and $f_1,...,f_{m-1}\in S_{fin}$. Furthermore, we write
$\mathcal{G}=\sum_{i=0}^{d-1} (\sum_{k=1}^m \mathcal{F}_k)^i$. First, a routine calculation leads to
$\mathcal{G}_{\sigma_0}=\mathcal{F}_{s_m}\circ\mathcal{G}\circ\mathcal{F}_{s_{m-1}}\cdots\mathcal{F}_{s_2}\circ\mathcal{G}\circ\mathcal{F}_{s_1}.$ Second, it follows from Theorem \ref{reachablem}  that for each $1\leq i\leq m-1$, and for any $\rho$, $\supp(\mathcal{F}_{f_i}(\rho))\subseteq \supp(\mathcal{G}(\rho)).$ Repeatedly applying this inclusion together with Lemma~\ref{prel-2}.1 we obtain
\begin{equation}\label{ft-3}\begin{split}
\supp(\mathcal{F}_{\sigma^\prime}(\rho))
=&\ \supp[(\mathcal{F}_{s_m}\circ\mathcal{F}_{f_{m-1}}\circ\mathcal{F}_{s_{m-1}}\circ\cdots\circ\mathcal{F}_{s_2}\circ\mathcal{F}_{f_1}\circ\mathcal{F}_{s_1})(\rho)]\\
\subseteq&\ \supp[(\mathcal{F}_{s_m}\circ\mathcal{G}\circ\mathcal{F}_{s_{m-1}}\circ\cdots\circ\mathcal{F}_{s_2}\circ\mathcal{G}\circ\mathcal{F}_{s_1})(\rho)] \\
=&\ \supp(\mathcal{G}_{\sigma_0}(\rho))\subseteq\supp(\sum_{\sigma\in \Pi}\mathcal{F}_\sigma)(\rho).\end{split}
\end{equation}

(4) Now we are able to complete the proof by showing that for any fair execution path $s\in F$, $s$ has an initial segment $t$ such that $\mathcal{F}_t(\rho_0)=0$. In fact, $s$ can be written as an infinite sequence of fair piece, i.e., $s=\sigma^\prime_1\sigma^\prime_2\cdots$,
where each $\sigma^\prime_i$ is a fair piece. We take $t$ to be the initial segment of $s$ containing the first $d$ fair pieces, i.e.,
$t=\sigma_1\sigma_2\cdots\sigma_d$. Repeatedly applying Eq.~(\ref{ft-3}) and Lemma~\ref{prel-2}.1 we obtain
\begin{equation*}\begin{split}\supp{\mathcal{F}_t(\rho_0)}=&\ \supp[(\mathcal{F}_{\sigma^\prime_{d}}\circ\cdots\circ\mathcal{F}_{\sigma^\prime_{2}}\circ\mathcal{F}_{\sigma^\prime_{1}})(\rho_0)]\\ \subseteq &\ \supp[(\sum_{p\in \Pi}\mathcal{F}_p)^d(\rho)]=\{0\}.\end{split}\end{equation*}
Thus, $\mathcal{F}_t(\rho)=0$. \hfill $\blacksquare$\end{proof}

The above theorem can be slightly strengthened by employing the notion of fairness index in Definition~\ref{ind-def}.
First, we have:
\begin{lem}\label{ind-lem} $\Pi^{\omega}\subsetneq F_{\frac{1}{md}}.$
\end{lem}

\begin{proof} For any $s=\sigma_1\sigma_2\cdots\in \Pi^{\omega}$ with $\sigma_i\in \Pi$,
we know that $\sigma_i(|\sigma_i|,k)\geq 1$ for any $k\in K$ and $|\sigma_i|<md$,
where $\sigma_i(|\sigma_i|,k)$ is the number of occurrences of $k$ in $\sigma_i$.
Then the occurrence frequency $f(s)>\frac{1}{md}$, which means that $\Pi^{\omega}\subseteq F_{\frac{1}{md}}$.
On the other hand, we choose an arbitrary $s\in F_{\frac{1}{md}}$.
Then $1^{md}s\in F_{\frac{1}{md}}$ but $1^{md}s\notin \Pi^{\omega}$.
\hfill $\blacksquare$
\end{proof}

Actually, what we proved in Theorem~\ref{ftermination} is that for any two schedules $A, B$ between $\Pi^\omega$ and $F$, i.e., $\Pi^{\omega}\subset A,B\subset F$, a program terminates in schedule
$A$ if and only if it terminates in schedule $B$ . Combining Theorem~\ref{ftermination} and Lemma~\ref{ind-lem} yields:
\begin{cor}
For any $0\leq \delta,\epsilon\leq\frac{1}{md}$,
a program terminates in schedule $F_{\delta}$ if and only if it terminates in schedule $F_{\epsilon}$. \hfill $\blacksquare$
\end{cor}

Now an algorithm checking termination of a concurrent quantum program can be developed based on Theorem~\ref{termination}.

\smallskip\
\begin{algorithm}
\caption{Decide termination of a concurrent quantum program \label{alg:TF}}
\SetKwInOut{Input}{input}\SetKwInOut{Output}{output}
\Input{An input state $\rho_0$, and the matrix representation of each $\F_i$
$i.e, N_i$}
\Output{b.(If the program terminates under $F$, $b=0$; Otherwise, $b=1$.)}
$N\leftarrow 0$\;
\For{$k=1:m$}{
$N\leftarrow N_i+N$\;
}
$G\leftarrow I$\;
\For{$k=1:d-1$}{
$G\leftarrow I+NG$\;
}
(*Compute the matrix representation of $\mathcal{G}$*)\\
$M\leftarrow 0$\;
Generate $P_K$\;
\For{$p=p_1p_2\cdots p_m\in P_K$}{
$L \leftarrow N_{p_1}$\;
\For{$l=2:m$}{
$L\leftarrow N_{p_l}GL$\;
}
(*Compute the matrix representation of $\mathcal{F}_p$*)\\
$M\leftarrow M+L$\;
}
(*Compute the matrix representation of $\sum_{p\in P_K}\mathcal{F}_p$*)\\
$\ket{x}\leftarrow M^d(\rho_0\otimes I)|\Phi\rangle$\;
\If{$\ket{x}\neq 0$}{
$b \leftarrow 1$;
}
\If{$\ket{x}= 0$}{
$b \leftarrow 0$;
}
\Return b
\end{algorithm}

\begin{thm}Algorithm 3 decides termination of a concurrent quantum program in time $\mathcal{O}(m^m~d^{4.7454})$, where $m$ is the number of the processes, and $d=\dim\hs.$\end{thm}

\textit{Proof:} In the algorithm, we use the for loop to compute the matrix representation $G$
of $\mathcal{G}=\sum_{i=0}^{d-1} (\sum_{k=1}^m \mathcal{F}_k)^i$.
Then the matrix representation of $\mathcal{F}_\sigma=\mathcal{F}_{s_1}\circ\mathcal{G}\circ\mathcal{F}_{s_2}\cdots\mathcal{G}\circ\mathcal{F}_{s_m}(\cdot)$
is obtained for any $\sigma=s_1s_2\cdots s_m\in P_K$.
All $\mathcal{F}_\sigma$s are added up to $M$. Then $M$ becomes the matrix representation of $\sum_{\sigma\in P_K}\mathcal{F}_\sigma$.
Consequently, we can apply Theorem \ref{ftermination} to assert that this algorithm outputs 0 if the program terminates in the fair schedule $F$; otherwise, 1.

To analyse its complexity, the algorithm can be divided into three steps: (1) Computing $G$ costs $\mathcal{O}(m+ d~d^{2*2.3727})=\mathcal{O}(m+d^{5.7454})$;
(2) Computing $M$ costs $m!*2m*\mathcal{O}(d^{2*2.3727})=\mathcal{O}(m^m~d^{4.7454})$;
(3) Computing $\ket{x}$ costs $\mathcal{O}(d^{4.7454}\log d)$.
So the total cost is $\mathcal{O}((m^m+d)d^{4.7454})$.\hfill $\blacksquare$
\section{Conclusion}\label{CCO}
In this paper, we studied two of the central problems, namely, reachability and termination for concurrent quantum programs. A concurrent quantum program is modeled by a family of quantum Markov chains sharing a state Hilbert space and a termination measurement, with each chain standing for a participating process. This model extends Hart, Sharir and Pnueli's model of probabilistic concurrent programs~\cite{HSP83} to the quantum setting. We show that the reachable space and the uniformly repeatedly reachable space of a concurrent quantum program can be computed and its termination can be decided in time $\mathcal{O}(d^{4.7454})$, $\mathcal{O}(d^{4.7454}\log d)$, $\mathcal{O}((m^m+d)d^{4.7454})$, respectively, where $m$ is the number of participating processes, and $d$ is the dimension of state space.

For further studies, an obvious problem is: how to improve the above algorithm complexities? In this paper, reachability and termination of quantum programs were defined in a way where probabilities are abstracted out; that is, only reachability and termination with certainty are considered. A more delicate, probability analysis of the reachability and termination is also an interesting open problem. The algorithms for computing the reachable space and checking termination of a \textit{quantum} program presented in this paper are all algorithms for \textit{classical} computers. So, another interesting problem is to find efficient \textit{quantum} algorithms for reachability and termination analysis of a quantum program.

\section*{Acknowledgment}
We are grateful to Dr Yangjia Li, Runyao Duan and Yuan Feng for useful discussions.
This work was partly supported by the Australian Research Council (Grant No. DP110103473).

\end{document}